\documentclass[a4paper,USenglish,cleveref,thm-restate]{lipics-v2021}
\usepackage[utf8]{inputenc}
\usepackage[T1]{fontenc}
\usepackage{amsmath,amssymb,amsthm}
\usepackage{tikz-cd}
\usepackage{tikz-3dplot}
\usetikzlibrary{math}
\usepackage{wrapfig}
\usepackage{macros}
\usepackage{fixes}

\pdfoutput=1 
\hideLIPIcs 

\title{Classifying covering types in homotopy type theory}

\author{Samuel Mimram}{LIX, CNRS, École polytechnique, Institut Polytechnique de Paris, 91120 Palaiseau, France}{samuel.mimram@polytechnique.edu}{https://orcid.org/0000-0002-0767-2569}{}
\author{Émile Oleon}{LIX, CNRS, École polytechnique, Institut Polytechnique de Paris, 91120 Palaiseau, France}{emile.oleon@polytechnique.edu}{https://orcid.org/0009-0001-8398-2577}{}
\authorrunning{S. Mimram and É. Oleon}
\Copyright{Samuel Mimram and Émile Oleon}



\category{} 

\relatedversion{} 

\nolinenumbers 

\EventEditors{Stefano Guerrini and Barbara K\"{o}nig}
\EventNoEds{2}
\EventLongTitle{34th EACSL Annual Conference on Computer Science Logic (CSL 2026)}
\EventShortTitle{CSL 2026}
\EventAcronym{CSL}
\EventYear{2026}
\EventDate{February 23--28, 2026}
\EventLocation{Paris, France}
\EventLogo{}
\SeriesVolume{363}
\ArticleNo{7}

\begin{document}
\maketitle

\begin{abstract}
  Covering spaces are a fundamental tool in algebraic topology because of the close relationship they bear with the fundamental groups of spaces. Indeed, they are in correspondence with the subgroups of the fundamental group: this is known as the \emph{Galois correspondence}. In particular, the covering space corresponding to the trivial group is the universal covering, which is a ``1-connected'' variant of the original space, in the sense that it has the same homotopy groups, except for the first one which is trivial. In this article, we formalize this correspondence in homotopy type theory, a variant of Martin-Löf type theory in which types can be interpreted as spaces (up to homotopy). Along the way, we develop an $n$-dimensional generalization of covering spaces. Moreover, in order to demonstrate the applicability of our approach, we formally classify the coverings of lens spaces and explain how to construct the Poincaré homology sphere.

\end{abstract}


\section{Introduction}
The notion of covering space is a fundamental tool in algebraic topology. It provides a canonical way to remove the low-dimensional homotopy structure of a space (the universal covering has a trivial fundamental group) and it bears a close relationship with the fundamental group: coverings are classified by subgroups of the fundamental group of the original space, which is known as the \emph{Galois correspondence}.
The setting of homotopy type theory~\cite{hottbook} allows one to perform geometric constructions in a synthetic way: all constructions on types correspond to manipulations of spaces, and are guaranteed to be invariant under homotopy of spaces by construction. It is thus natural to expect that the definition of covering space and associated properties can be developed in this framework, and we explain here that this is indeed the case. The notion of (universal) covering of a type was first introduced by Harper and Favonia in~\cite{harper2018covering}, and the Galois correspondence was recently independently shown by Wemmenhove, Manea, and Portegies~\cite{wemmenhove2024classification}. Here, we develop further the theory of covering spaces, by explaining their relationship with the connected/truncated factorization, generalizing to $n$-coverings (we recover the usual case by setting $n=0$), and computing their homotopy groups. Compared to~\cite{wemmenhove2024classification}, our proof of the Galois correspondence departs from the traditional one in algebraic topology~\cite{hatcher}, providing arguments which are shorter, more conceptual, and should be amenable to generalizations (in particular, we leave the general classification of $n$-coverings for future work, handling only the case $n=0$ here). As a novel, important and non-trivial application of our constructions, we provide a classification of the covering spaces of lens spaces, which play an important role in algebraic topology because they provide deloopings of cyclic groups. We also more briefly present other applications such as the construction of important spaces due to Poincaré (the hypercubical manifold and the homology sphere) as quotients of coherent actions of their fundamental group on their universal covering.

\subparagraph{Plan of the paper.}
After recalling basic notations and concepts in homotopy type theory (\cref{hott}), we define and study higher covering types (\cref{higher-cov}), and prove the classification of covering types (\cref{galois-correspondence}). Finally, as an application, we classify the covering types of lens spaces, construct the Poincaré homology sphere (\cref{applications}) and conclude (\cref{conclusion}).

\ifx\authoranonymous\relax\else
\subparagraph{Acknowledgments.}
We would like to thank Hugo Salou, who has formalized our proof of \cref{covering-classification} in cubical Agda~\cite{covering-agda}, for useful comments and discussions.
\fi

\section{Homotopy type theory}
\label{hott}
We begin by briefly recalling the main notations and tools of homotopy type theory. Detailed presentations can be found in~\cite{hottbook,rijke2022introduction}.

\subparagraph{Universe.}
We write $\U$ for the \emph{universe} type, whose elements are small types, which is supposed to be closed under dependent sum and product types.
Given a type $A:\U$ and a type family $B:A\to\U$, we write $\Sigma A.B$ or $\Sigma(x:A).B\,x$ for dependent sum types and $\Pi A.B$ or $\Pi(x:A).B\,x$ or $(x:A)\to B\,x$ for dependent product types. As customary, we respectively write $A\times B$ and $A\to B$ for product and arrow types, which correspond to non-dependent particular cases of the previous constructions. We respectively write~$0$ and~$1$ for the initial and terminal types.

\subparagraph{Identities.}
The type theory features a notion of definitional equality and we write $t\eqdef u$ when two terms $t$ and $u$ are definitionally equal. It also features a notion of propositional equality: given $x,y:A$, we write $x=y$ for the type of \emph{identities} or \emph{paths} between~$x$ and~$y$ in~$A$. For any point $x:A$, there is a path $\refl:x=x$ witnessing reflexivity. Given paths $p:x=y$ and $q:y=z$, we can construct paths $p\cdot q:x=z$ corresponding to concatenation or transitivity, and $\sym p:y=x$ corresponding to inverse or symmetry.

\subparagraph{Pointed types.}
A \emph{pointed type} is a type $A$ together with a distinguished element, often written $\pt_A$ or even $\pt$ (equivalently, the distinguished element can be specified by providing a map $1\to A$).
A \emph{pointed} map $f:A\to B$ between pointed types~$A$ and~$B$ is a map between the underlying types together with an identification $f(\pt_A)=\pt_B$. We write $A\pto B$ for the corresponding type of pointed maps.

\subparagraph{Univalence.}
A map $f:A\to B$ is an \emph{equivalence} when it admits both a left and a right inverse. We write $A\equivto B$ for the type of equivalences between~$A$ and~$B$. Any identity between two types canonically induces an equivalence between them. The \emph{univalence} axiom states that the corresponding map $(A=B)\to(A\equivto B)$ is itself an equivalence: in particular, any equivalence induces an identity.

\subparagraph{Homotopy levels.}
A type~$A$ is \emph{contractible} when it is equivalent to~$1$.
A type~$A$ is a \emph{proposition} (\resp a \emph{set}, \resp a \emph{groupoid}) when for any elements $x,y:A$ the type $x=y$ is contractible (\resp a proposition, \resp a set).
The type of sets is denoted $\Set$.
More generally, we can define a notion of $n$-type for $n\in\set{-2,-1}\cup\N$ by stating that a $(-2)$-type is a contractible one, and an $(n{+}1)$-type is a type~$A$ in which $x=y$ is an $n$-type for every~$x,y:A$. We write $\isType n(A)$ for the predicate indicating that~$A$ is an $n$-type, which can be shown to be a proposition. We write $\Type[n]\defd\Sigma\U.\isType n$ for the type of $n$-types.

\subparagraph{Truncation.}
The \emph{$n$-truncation} $\trunc nA$ is the universal way of turning a type~$A$ into an $n$-type: it comes equipped with a map $\trunq n{{-}}:A\to\trunc nA$ such that any map $f:A\to B$, whose target $B$ is an $n$-type, induces a unique map $\tilde f:\trunc nA\to B$ such that $\tilde f(\trunq nx)=f(x)$ for $x:A$.
\label{n-connected}
A type~$A$ is \emph{$n$-connected} when $\trunc nA=1$. In particular, a type is \emph{connected} (\resp \emph{simply connected}) when it is $0$-connected (\resp $1$-connected).
The \emph{connected component} of an element $a$ of~$A$ is $\Sigma(x:A).\ptrunc{a=x}$.

\subparagraph{Loop space.}
The \emph{circle} $\S1$ is the free pointed type containing a path $\Sloop:\pt=\pt$.
Given a pointed type~$A$, its \emph{loop space}~$\Loop A$ is $\pt=\pt$, which can be shown to coincide with the type $\S1\pto A$. Its \emph{fundamental group} is the type $\pi_1(A)\defd\strunc{\Loop A}$, which is canonically equipped with a group structure induced by path concatenation.
When $A$ is a pointed connected groupoid, its loop space coincides with its fundamental group, so that it is a group. A \emph{delooping} of a group~$G$ is a type $\B G$ equipped with an isomorphism of groups $\Loop\B G\isoto G$ (such a type always exists and is unique, thus the notation). It is easily shown that a map $\B G\to\Set$ corresponds to a set equipped with an \emph{action} of~$G$ in the usual sense. 

\subparagraph{Fiber sequences.}
Given a map $f:A\to B$ and $y:B$, the \emph{fiber} of $f$ at $y$ is the type $\fib fy\defd\Sigma(x:A).(f(x)=y)$. When~$B$ is pointed, the \emph{kernel} of~$f$ is $\ker f\defd\fib f\pt$. A composable pair of morphisms $
\begin{tikzcd}[cramped,sep=small]
  F\ar[r]&A\ar[r,"f"]&B
\end{tikzcd}
$ is a \emph{fiber sequence} when $F$ is the kernel of~$f$ and the map $F\to A$ is the first projection.
A map is $n$-connected (\resp $n$-truncated) when all its fibers are.

\subparagraph{The Grothendieck duality.}
A fundamental property in homotopy type theory is that, given a type~$A$, \emph{fibrations} over~$A$ correspond both to types over~$A$ and to type families indexed by~$A$: this is the \emph{Grothendieck duality}. We write $\Type/A\defd\Sigma(B:\U).(B\to A)$ for the type of \emph{types over $A$} and $A\to\U$ for the type of \emph{type families} indexed by~$A$. We have a function $\fib{}:\U/A\to(A\to\U)$, which to $p:B\to A$ associates the fiber $\fib p:A\to\U$, and a function $\int:(A\to\U)\to(\U/A)$, which to a family $F:A\to\U$ associates the first projection map $\fst:\Sigma A.F\to A$. The following is shown in~\cite[Section 4.8]{hottbook}:

\begin{theorem}
  \label{grothendieck-duality}
  Given a type~$A$, the above functions induce an equivalence
  $
  \U/A\equivto(A\to\U)
  $.
  Moreover, this correspondence is functorial in the sense that
  given $p:B\to A$ and $q:C\to A$, a morphism $f:B\to C$ with $q\circ f=p$ corresponds to a family of maps
  $
  (x:A)\to\fib px\to\fib qx
  $
  in a way which preserves identities and composition (and thus equivalences).
\end{theorem}








\section{Higher covering types}
\label{higher-cov}

\subsection{Covering spaces in topology}
\label{cov-top}
We briefly recall here the traditional notion of covering space in topology and refer to standard textbooks for details~\cite{hatcher}. Given a topological space~$A$, a \emph{covering} is a space~$B$ together with a map~$p:B\to A$ which is \emph{locally trivial}. This means that every point~$x:A$ admits an open neighborhood such that the preimage $p^{-1}(U)$ is homeomorphic to a space of the form $U\times F$, where~$F$ is a set equipped with the discrete topology and the restriction of $p$ to~$p^{-1}(U)$ is the first projection. When~$A$ is connected, which will be the case of interest here, the cardinality of~$F$ has to be the same for every point~$x$ and is called the number of \emph{sheets} of the covering. Below we have figured a covering with 3 sheets (on the left) and with countably many sheets (on the right):
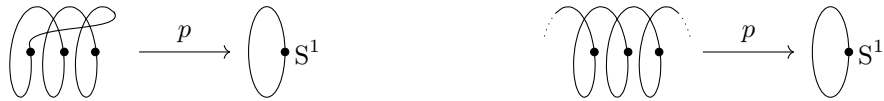
\[
  \begin{tikzpicture}[scale=1.2,declare function={f(\x)=0.2*sin(\x)+\x/1000;},rotate=90,xscale=0.5,every circle node/.style={draw,fill,inner sep=1}]
    \draw plot[variable=\x,domain=-30:1000,samples=55,smooth] ({cos(\x)},{f(\x)}) to[out=0,in=195] cycle;
    \draw (0,-2) arc(-90:270:1cm and 0.2cm);
    \draw[->]  (0,-0.4) -- (0,-1.4) node[midway,above]{$p$};
    \draw (0,.79) node[circle]{};
    \filldraw (0,.42) node[circle]{};
    \filldraw (0,.08) node[circle]{};
    \filldraw (0,-2) node[circle]{} node[right]{$\S1$};
  \end{tikzpicture}
  \qquad\qquad\qquad\qquad
  \begin{tikzpicture}[scale=1.2,declare function={f(\x)=0.2*sin(\x)+\x/1000;},rotate=90,xscale=0.5,every circle node/.style={draw,fill,inner sep=1}]
    \filldraw[white] (0,1) circle (0.02);
    \draw plot[variable=\x,domain=-30:1110,samples=60,smooth] ({cos(\x)},{f(\x)});
    \draw[dotted] plot[variable=\x,domain=-70:-30,samples=10,smooth] ({cos(\x)},{f(\x)});
    \draw[dotted] plot[variable=\x,domain=1110:1150,samples=10,smooth] ({cos(\x)},{f(\x)});
    \draw (0,-2) arc(-90:270:1cm and 0.2cm);
    \draw[->]  (0,-0.4) -- (0,-1.4) node[midway,above]{$p$};
    \filldraw (0,.79) node[circle]{};
    \filldraw (0,.42) node[circle]{};
    \filldraw (0,.08) node[circle]{};
    \filldraw (0,-2) node[circle]{} node[right]{$\S1$};
  \end{tikzpicture}
\]
In order for the constructions of coverings to be well-behaved, we will implicitly assume in the following that~$A$ satisfies reasonable assumptions (namely being connected, locally path-connected, semilocally simply-connected). We will also assume that it comes equipped with a distinguished point~$\pt$. An important feature of coverings is that they have the \emph{path lifting property}: given a path $\pi:x\rightsquigarrow y$ in~$A$ and an element $\tilde x\in\tilde A$ with $p(\tilde x)=x$ there is a unique path $\tilde\pi:\tilde x\to\tilde y$ with $p(\tilde\pi)=\pi$. This implies that we have an action of the fundamental group~$\pi_1(A)$ on the fiber $p^{-1}(\pt)$ and, in fact, this characterizes coverings:

\begin{proposition}[{\cite[Section 1.3, p. 70]{hatcher}}]
  \label{covering-action}
  Coverings of~$A$ are in bijection with sets equipped with an action of~$\pi_1(A)$.
\end{proposition}

The \emph{universal covering} is the only covering $p:\tilde A\to A$ with $\tilde A$ simply connected. For instance, the universal covering of $\S1$ is the ``helix'' pictured on the right above. This space can be shown to be unique up to isomorphism and can be constructed as follows:

\begin{proposition}[{\cite[Section 1.3, p. 64]{hatcher}}]
  \label{top-universal-cover}
  The universal cover $\tilde A$ of~$A$ can be constructed as the space whose points are pairs consisting of a point~$x:A$ and a path $p:\pt\rightsquigarrow x$ up to homotopy, equipped with a suitable topology, the map $p:\tilde A\to A$ being given by first projection.
\end{proposition}

The \emph{fundamental group} $\pi_1(A)$ of~$A$ is the group whose elements are homotopy classes of paths $\pt\rightsquigarrow\pt$ in~$A$, with concatenation as multiplication and constant paths as neutral elements. By the construction of \cref{top-universal-cover}, given $x:A$, the fiber $p^{-1}(x)$ is~$\pi_1(A)$, and we have an action of $\pi_1(A)$ on this fiber given by left multiplication, which induces an action of $\pi_1(A)$ on~$\tilde A$.

\begin{proposition}[{\cite[Theorem 1.38]{hatcher}}]
  \label{ucov-quotient}
  The quotient $\tilde A/\pi_1(A)$ of the universal cover under the above action is precisely~$A$.  
\end{proposition}

\noindent
The action can be shown to be free, so that the quotient above coincides with the homotopy quotient.



\subsection{The universal fibration}
\label{universal-fibration}
We begin by describing a situation in homotopy type theory, which is close to coverings, and fundamental with respect to the characterization of identity types. This construction might seem a bit artificial at first, but we will see that the point of view nicely generalizes to coverings, in the sense that what we describe in this section is a kind of ``universal $\infty$-covering''.

Suppose given a pointed type~$A$. A \emph{pointed type over $A$} is a pointed morphism $p:B\pto A$. A \emph{morphism} between two such pointed types $p:B\pto A$ and $q:C\pto A$ is a map $f:B\to C$ together with an equality $f\circ p=q$. Such morphisms compose in the expected way, which is compatible with the composition of underlying maps.

\begin{definition}
  \label{unviersal-pointed-type}
  A pointed type~$p:B\pto A$ is \emph{universal} when for every pointed type~$q:C\to A$, the type of morphisms from~$p$ to~$q$ is contractible.
\end{definition}

\noindent
Since it satisfies a universal property, the universal pointed type is unique, and it exists thanks to the following characterization:

\begin{proposition}
  \label{contractible-univ-oo-covering}
  A pointed type~$p:B\pto A$ is universal if and only if~$B$ is contractible.
\end{proposition}
\begin{proof}
  By definition of morphisms and initiality of~$1$, the type $1$ is initial among pointed types over $A$. A universal pointed type $p:B\to A$ is also initial among pointed types over $A$, by the universal property, and thus $B$ is equivalent to $1$, \ie contractible.
\end{proof}




\noindent
By immediate computations, we have:

\begin{proposition}
  The fiber of the universal pointed type~$p:1\to A$ is $\Loop A$.
\end{proposition}

It is very illuminative to translate the previous definitions and properties under the Grothendieck duality (see \cref{grothendieck-duality}). A pointed type over~$A$ corresponds to a fibration $P:A\to\U$ together with a distinguished element~$\pt_P$ of~$P\pt$, and, by \cref{contractible-univ-oo-covering}, such a fibration is universal precisely when the total space $\Sigma A.P$ is contractible. The universal pointed type thus corresponds to the \emph{universal fibration}, which is the map $\UF A:A\to\U$ defined by $\UF Ax\eqdef(\pt=x)$, and pointed by $\refl[\pt]$.
By \cref{grothendieck-duality}, a morphism between fibrations~$P,Q:A\to\U$ corresponds to a family of maps
$
  f:(x:A)\to P\,x\to Q\,x
$
together with an identification $f\pt\pt_P=\pt_Q$. 
The initiality property of universal pointed types (\cref{unviersal-pointed-type}) then translates as the \emph{fundamental theorem of identity types}~\cite[Theorem 11.2.2]{rijke2022introduction}:

\begin{theorem}
  Suppose given a type family~$P:A\to\U$ pointed by $\pt_P:P\,\pt$, together with a family of maps
  \[
    F:(x:A)\to(\pt=x)\to P\,x
  \]
  and an identification $F\,\pt\,\refl[\pt]=\pt_P$. Then $F$ is a family of equivalences if and only if the total space $\Sigma A.P$ is contractible.
\end{theorem}
\begin{proof}
  Under the Grothendieck duality (\cref{grothendieck-duality}), the type family~$P$ corresponds to the pointed type $p:\Sigma A.P\to A$ over~$A$, given by the first projection. The family of maps~$F$ then corresponds to a morphism between the universal type over~$A$ and~$p$, \ie to a pointed map $f:1\to\Sigma A.P$. Then~$F$ is a family of equivalences if and only if the induced map~$f$ is an equivalence, \ie if and only if $\Sigma A.P$ is contractible.
\end{proof}

\begin{remark}
  In the situation of the above theorem, the map~$P$ can be thought of as being \emph{representable}, in the sense that $Px=(\pt=x)$ (and identities can be thought of as homs in types). This notion was actually used by Voevodsky~\cite{voevodsky2006very} in order to first define identity types.
\end{remark}



\subsection{Higher covering types}
We now introduce the notion of~$n$-covering type, for any natural number~$n$, which can be understood as an $n$-truncated variant of the pointed types of previous section. The traditional notion of covering is the particular case $n=0$, as we indicate in remarks.

An \emph{$n$-covering} of~$A$ is a map $p:B\to A$ whose fibers are $n$-types; such a map is also said to be \emph{$n$-truncated}. We write
\[
  \Covering[n](A)
  \qdefd
  \Sigma(B:\U).\Sigma(f:B\to A).(x:A)\to\isType n(\fib fx)
\]
for the type of coverings of~$A$. Under the Grothendieck duality, those can also be defined as families of $n$-types.


\begin{lemma}
  \label{covering-maps}
  We have an equivalence $\Covering[n](A)\equivto(A\to\Type[n])$.
\end{lemma}
\begin{proof}
  Follows immediately from Grothendieck duality (\cref{grothendieck-duality}).
\end{proof}

\begin{remark}
  \label{cov-pi-act}
  For $n=0$, we recover the definition of covering types of~\cite[Definition 1]{harper2018covering}, as maps~$A\to\Set$.
  Since $\Set$ is a groupoid~\cite[Theorem 7.1.11]{hottbook}, the universal property of groupoid truncation provides us with an equivalence $(A\to\Set)\equivto(\gtrunc A\to\Set)$. A covering of~$A$ thus corresponds to a map $\gtrunc{A}\to\Set$. Moreover, we have by \cref{gtrunc-Bpi1} that $\gtrunc{A}$ is a $\B\pi_1A$, from which we deduce that a covering of~$A$ corresponds to a set equipped with an action of~$\pi_1A$. We thus recover the traditional \cref{covering-action}, which was formalized in homotopy type theory in~\cite[Theorem 4]{harper2018covering}.
\end{remark}

\noindent
A \emph{morphism} $f$ between $n$-coverings $p:B\to A$ and $q:C\to A$ is a map $f:B\to C$ together with an equality $p=q \circ f  $.

A \emph{pointed} $n$-covering is a pointed map $p:B\pto A$ between pointed types, whose underlying map is an $n$-covering. This corresponds exactly to the following notion:

\begin{definition}
  \label{n-pcov}
  A \emph{pointed $n$-covering} is a factorization
  \[
    \begin{tikzcd}
      1\ar[dr,"a"']\ar[rr,"i"]&&B\ar[dl,"p"]\\
      &A
    \end{tikzcd}
  \]
  of the pointing map $a:1\to A$ as $a=p\circ i$ where $p$ is $n$-truncated.
\end{definition}

\noindent
In the following, we sometimes assimilate the covering to the map $p:B\to A$ leaving the data of the pointing $i:1\to B$ of~$B$ implicit.
As expected, a morphism $f$ between pointed $n$-coverings $p$ and $q$ is a pointed map which is a morphism between the underlying $n$-coverings, \ie such that the factorization of the target is $q\circ(f\circ i)$:
\[
  \begin{tikzcd}
    1\ar[dr,"a"']\ar[r,"i"]&B\ar[d,"p"]\ar[r,"f"]&C\ar[dl,"q"]\\
    &A
  \end{tikzcd}
\]

\begin{definition}
  \label{n-ucov}
  A pointed $n$-covering $p:B\to A$ is \emph{universal} when for every $n$-covering $q:C\to A$ there exists a unique map $f:B\to C$ of pointed $n$-coverings.
\end{definition}

\begin{remark}
  For $n=0$, we recover the characterization of pointed universal coverings as being initial in the category of pointed coverings~\cite[Lemma 12]{harper2018covering}.
\end{remark}

\noindent
We sometimes write $\tilde A$ for the universal $n$-covering. This is justified by the fact that, being defined by a universal property, it is uniquely characterized:

\begin{lemma}
  Any two universal pointed $n$-coverings are uniquely isomorphic.
\end{lemma}

We have the following characterization of universal $n$-coverings:

\begin{theorem}
  \label{n-ucov-conn}
  Given a pointed type $A$ with pointing map~$a:1\to A$, any factorization $a=p\circ i$ as $n$-connected map $i$ and followed by an $n$-truncated map~$p$ exhibits~$p$ as a universal $n$-covering. Moreover, such a factorization always exists and is unique.
\end{theorem}
\begin{proof}
  By~\cite[Theorem 7.6.6]{hottbook}, the map $a:1\to A$ admits a unique factorization $a=p\circ i$ as required. Moreover, given a $n$-covering $q:B\to A$ pointed by $b:1\to B$, we have a commuting square as on the left
  \[
    \begin{tikzcd}
      1\ar[d,"i"']\ar[r,"b"]\ar[dr,"a"]&B\ar[d,"q"]\\
      \tilde A\ar[r,"p"']&A
    \end{tikzcd}
    \qquad\qquad\qquad\qquad
    \begin{tikzcd}
      1\ar[d,"i"']\ar[r,"b"]&B\ar[d,"q"]\\
      \tilde A\ar[r,"p"']\ar[ur,dotted]&A
    \end{tikzcd}
  \]
  The commutation of the two triangles being given by the fact that we have two pointed $n$-coverings of~$A$ pointed by~$a$. By \cite[Theorem 7.6.7]{hottbook}, since $i$ is $n$-connected and $p$ is $n$-truncated, there is a unique map $\tilde A\to B$ making the two triangles on the right commute, and $p$ is thus universal in the sense of \cref{n-ucov}.
\end{proof}

In practice, the universal $n$-covering can be constructed as follows. We recall from~\cite[Definition 7.6.3]{hottbook} that, given a map $f:B\to A$ and a natural number~$n$, its \emph{$n$-image} is
\[
  \im_nf
  \quad\defd\quad
  \Sigma(x:A).\trunc n{\fib fx}
\]
We write $i_n:B\to\im_nf$ for the canonical map such that $i_n(x)=(f(x),\trunq n{\refl[x]})$ and $p_n:\im_nf\to A$ for the first projection map.

\begin{proposition}
  \label{ucov-image}
  The factorization $a=p_n\circ i_n$ of the pointing map~$a$ as above, exhibits $p_n:\im_na\to A$ as the universal covering of $A$.
\end{proposition}
\begin{proof}
  By \cite[Lemma 7.6.4]{hottbook}, this factorization satisfies the conditions of \cref{n-ucov-conn}.
\end{proof}

\begin{remark}
  \label{ucov-hclasses}
  For $n=0$, we recover the usual definition of the universal covering space as the type of homotopy classes of paths from the distinguished point:
  \[
    \tilde A
    \quad=\quad
    \Sigma(x:A).\strunc{\pt=x}
  \]
  For $n=-1$, the universal covering of~$A$ is its set of connected components. Finally, for $n=\infty$ (we adopt the convention that $\trunc\infty A\eqdef A$), we recover the universal pointed type of \cref{universal-fibration} by contractibility of singletons~\cite[Lemma 3.11.8]{hottbook}.
\end{remark}

\begin{example}
  The universal covering of $A\defd\S1$ is $\Sigma(x:A).(\pt=x)$ (we can remove the 0-truncation because $\S1$ is a groupoid) and thus contractible by~\cite[Lemma 3.11.8]{hottbook} (see \cref{contractible-ucov} for a generalization of this argument).
\end{example}

\begin{remark}
  The factorization results used above hold more generally for any map $a:B\to A$ where $B$ is not necessarily contractible. In this sense, given an arbitrary map $a:B\to A$, we can think of $\im_na$ as the ``universal $n$-cover of~$A$ relative to~$a$''.
\end{remark}

The universal $n$-covering can also be characterized as the covering whose total space is $(n{+}1)$-connected. This can be shown using the following lemma proved in appendix.

\begin{lemma}
  \label{n+1-connected}
  \label{n+1-truncated}
  A pointed connected type~$A$ is $(n{+}1)$-connected (\resp $(n{+}1)$-truncated) if and only if the pointing map $1\to A$ is $n$-connected (\resp $n$-truncated).
\end{lemma}
\begin{proof}
  First, consider a property $P:(x,y:A)\to(p:x=y)\to\U$ such that $P\,x\,y\,p$ is a proposition for any $x,y:A$ and $p:x=y$. Since $A$ is pointed connected, the following are equivalent
  \begin{enumerate}[(i)]
  \item $P\,x\,y\,p$ for any $x$, $y$ and $p$,
  \item $P\pt\pt\,p$ holds for any $p:\pt=\pt$.
  \end{enumerate}
  Indeed, the second is a particular case of the first. Conversely, suppose (ii) holds, and that we are given $x$, $y$, and $p$. Since we want to show $P\,x\,y\,p$ which is a proposition, we can suppose given paths $p_x:\pt=x$ and $p_y:\pt=y$ because $A$ is connected, we then deduce that $P\,x\,y\,p$ holds from (ii) by transport.

  Now, we have that $A$ is $(n{+}1)$-connected if and only if $x=y$ is $n$-connected for any $x,y:A$ (see \cref{n-connected}) if and only if $\Loop A$ is $n$-connected (by the preceding observation, based on the fact that being $n$-connected is a proposition \cite[Theorem 7.1.10]{hottbook}) if and only if the pointing map $a=1\to A$ is $n$-connected (because we have $\Loop A=\fib a\pt$). The reasoning is similar for the truncated version.

  An alternative proof for the connected case can be found in~\cite[Lemma 7.5.11]{hottbook}.
\end{proof}

\begin{theorem}
  \label{n-cov-conn}
  Given a pointed connected type $A$, the universal pointed $n$-covering is the $(n{+}1)$-connected pointed $n$-covering of~$A$.
\end{theorem}
\begin{proof}
  By \cref{n-pcov}, a pointed $n$-covering~$p:B\to A$ is a factorization of the pointing map $a:1\to A$ as $a=p\circ i$ with $p$ $n$-truncated. By \cref{n-ucov-conn}, it is universal if and only if $i:1\to B$ is $n$-connected which, by \cref{n+1-connected} is equivalent to the fact that $B$ is $(n{+}1)$-connected. 
\end{proof}

\begin{remark}
  \label{ucov-1con}
  For $n=0$, we recover the fact that the universal covering is the only 1-connected covering of a type~\cite[Lemma 11]{harper2018covering}.
\end{remark}

We now formalize the intuition that the universal $n$-covering $\tilde A$ of~$A$ provides a way to ``kill'' all the homotopy in dimension $i\leq n+1$. This is based on what we call the \emph{fundamental fibration} associated to the universal $n$-covering:

\begin{theorem}
  \label{ucov-fiber-sequence}
  Writing $\tilde A$ for the universal $n$-covering, we have a fiber sequence
  \[
    \begin{tikzcd}[sep=large]
      \tilde A\ar[r,"p"]&A\ar[r,"\trunq{n+1}{{-}}"]&\trunc{n+1}A
    \end{tikzcd}
  \]
\end{theorem}
\begin{proof}
  We have
  \[
    \ker\trunq{n+1}{{-}}
    \quad\eqdef\quad
    \Sigma(x:A).(\trunq{n+1}{\pt}=\trunq{n+1}{x})
    \quad=\quad
    \Sigma(x:A).\trunc n{\pt=x}
    \quad=\quad
    \tilde A
  \]
  where middle equality is \cite[Theorem 7.3.12]{hottbook} and right one is \cref{ucov-image}.  
\end{proof}

\noindent
This was actually taken to be the definition of $n$-coverings in~\cite{buchholtz2023long}.
As a corollary, we have the following characterization of the homotopy groups of the universal $n$-covering:

\begin{proposition}
  \label{ucov-pi}
  We have $\pi_i(\tilde A)=1$ for $i\leq n+1$ and $\pi_i(\tilde A)=\pi_i(A)$ for $i>n+1$.
\end{proposition}
\begin{proof}
  Since $\tilde A$ is $(n{+}1)$-connected by \cref{n-cov-conn}, we have $\pi_i(\tilde A)=1$ for $i\leq n+1$~\cite[Lemma 8.3.2]{hottbook}. For $i>n+1$, this is a consequence of the long exact sequence induced by the fiber sequence of \cref{ucov-fiber-sequence}, see~\cite[Theorem 8.4.6]{hottbook}:
  \[
    \begin{tikzcd}
      \cdots\ar[r]&\pi_{i+1}\trunc{n+1}A\ar[r]&\pi_i\tilde A\ar[r]&\pi_iA\ar[r]&\pi_i\trunc{n+1}A\ar[r]&\pi_{i+1}\tilde A\ar[r]&\cdots
    \end{tikzcd}
  \]
  We have $\pi_i\trunc{n+1}A=1$ and the short exact sequence
  \[
    \begin{tikzcd}
      1\ar[r]&\pi_i\tilde A\ar[r]&\pi_iA\ar[r]&1
    \end{tikzcd}
  \]
  shows that $\pi_i\tilde A=\pi_iA$.
\end{proof}

\noindent
In particular, this suggests that the universal covering should be contractible when $A$ has no homotopy in dimension $i>n+1$:

\begin{lemma}
  \label{contractible-ucov}
  Given a $(n{+}1)$-truncated pointed type~$A$, its universal $n$-covering is contractible.
\end{lemma}
\begin{proof}
  Since~$A$ is supposed to be $(n{+}1)$-truncated, the pointing map $a:1\to A$ is $n$\nbd-trunca\-ted by \cref{n+1-truncated} and the factorization $a=a\circ\id{1}$
  of the pointing map as an $n$-connected map followed by an $n$-truncated map has to be the factorization of the universal covering (\cref{n-ucov-conn}) by uniqueness.
\end{proof}

We would like to end this section with the following conjecture, which would allow constructing deloopings of higher groups based on the previous construction. Note that, below, the coverings are not supposed to be pointed.

\begin{conjecture}
  Given a connected type~$A$, the connected component of $\tilde A$ in $\Covering[n](A)$ is $\trunc{n+1}A$.
\end{conjecture}

\noindent
Provided that this conjecture holds, the connected component of $\tilde A$ in~$A$ would be a pointed $(n{+}1)$-connected groupoid, which can thus be thought of as a delooping of the fundamental $n$-group of~$A$. The following proposition shows that we have right underlying type (it would remain to be shown that we have the right higher operations for the $n$-group):

\begin{proposition}
  \label{aut-ucov}
  Given a pointed connected type~$A$, the type of automorphisms of the universal $n$-covering~$\tilde A$ is $\trunc n{\Loop A}$.
\end{proposition}
\begin{proof}
  By the universal property of $\tilde A$ (\cref{unviersal-pointed-type}), an automorphism $f:\tilde A\to\tilde A$ is uniquely determined by $f\pt$ which is an element of $\fib p\pt$.
  The type of automorphisms of $\tilde A$ is thus $\fib p\pt$. Now, by \cref{ucov-image}, we can consider that $p$ is the first projection $p:\Sigma(x:A).\trunc n{\fib a x}\to A$ whose fiber at~$\pt$ is $\trunc n{\fib a\pt}$ by \cite[Lemma 4.8.1]{hottbook}, \ie $\trunc n{\Loop A}$.
\end{proof}

\begin{remark}
  Let us explain why the conjecture does hold in the case $n=0$. Writing
  \[
    \Comp(p_0)
    \qeqdef
    \Sigma(p:\Covering(A)).\ptrunc{p_0=p}
  \]
  for the connected component of the universal 0-covering $p_0:\tilde A\to A$, we have
  \[
    \Loop_{p_0}\Comp(p_0)
    =
    \Loop_{p_0}\Covering(A)
    =
    \strunc{\Loop A}
    \eqdef
    \pi_1(A)
  \]
  where the first equality follows from the fact that the canonical projection map from $\Comp(p_0)$ to $\Covering(A)$ is an embedding~\cite[Lemma 7.6.4]{hottbook}, and the second one is due to \cref{aut-ucov}. Moreover, this identity is compatible with the group structures on both sets.

  Given a group~$G$, consider a delooping $A\defd\B G$. By Grothendieck duality, the type of coverings of $\B G$ coincide with maps $\B G\to\Set$ (see \cref{cov-pi-act}), \ie with the type $\GSet$ of sets equipped with an action of~$G$~\cite[Theorem 4]{harper2018covering}.
  Moreover, under this identification, the universal covering corresponds to the principal $G$-set $P_G$, which is the set~$G$ equipped with the canonical action induced by right multiplication. Indeed, $\B G$ being a $1$-truncated pointed type, its universal covering is contractible by \cref{contractible-ucov} and is thus the pointing map~$p:1\to\B G$, and the corresponding map $\phi:\B G\to\Set$ is $\phi\eqdef\fib p\eqdef(x\mapsto\pt=x)$. In particular, we have $\phi(\pt)=(\pt=\pt)=\Loop\B G=G$. Moreover, given an element~$a$ of~$G$, seen as path $a:\pt=\pt$, we have $\transport{(\ap{\phi}a)}(b)=ab$ by the formula for transport in identity types~\cite[Lemma 2.11.2]{hottbook}. We thus have that the connected component of the principal $G$-set
  \[
    \Comp(P_G)
    \eqdef
    \Sigma(X:\GSet).\ptrunc{P_G=X}
  \]
  is a delooping of~$G$. This type is known as the type of \emph{$G$-torsors}~\cite{symmetry,generated,demazure1970groupes,warn2023eilenberg}.
\end{remark}


\section{The Galois correspondence}
\label{galois-correspondence}

\subsection{The Galois fibration}
From now on, we restrict ourselves to the case $n=0$ of $n$-coverings.
In order to define the Galois correspondence, we first need to define the action of the fundamental group $\pi_1(A)$ of a pointed connected type~$A$ on its universal cover~$\tilde A$. This means that we want to define a map $F:\B\pi_1A\to\U$ such that $F\pt=\tilde A$. By the Grothendieck duality (\cref{grothendieck-duality}), this amounts to define a map $f:X\to\B G$, for some type~$X$, whose fiber is $\tilde A$. Moreover, the source~$X$ has to be the homotopy quotient of $\tilde A$ under this action, which is known to coincide with the strict quotient because the action is free, and should thus be $A$ itself by \cref{ucov-quotient}, see \cref{cov-top}. Another important observation, is that we have a very convenient model for~$\B\pi_1A$, namely:

\begin{proposition}
  \label{gtrunc-Bpi1}
  Given a pointed connected type~$A$, we have that $\gtrunc A$ is a $\B\pi_1A$.
\end{proposition}
\begin{proof}
  We take $\gtrunq\pt$ to be the distinguished point of $\gtrunc A$. Connectedness is preserved by truncation: we have $\strunc{\gtrunc A}=\strunc{A}=1$ (the first equality is \cite[Lemma 7.3.15]{hottbook} and the second one is the fact that $A$ is connected) and thus $\gtrunc A$ is connected. Finally, we have $\Loop\gtrunc A=\strunc{\Loop A}\eqdef\pi_1 A$ by \cite[Corollary 7.3.13]{hottbook}.
\end{proof}

\noindent
The previous discussion suggests defining:

\begin{definition}
  Given a pointed connected type~$A$, the associated \emph{Galois fibration} is the map $\gtrunq{{-}}:A\to\gtrunc A$, which we write $g_A$ in the following.
\end{definition}

\noindent
Its target is $\B\pi_1A$ by \cref{gtrunc-Bpi1} and we have the expected fiber as the case $n=0$ of \cref{ucov-fiber-sequence}:

\begin{proposition}
  \label{universal-covering-fiber-sequence}
  We have $\ker g_A=\tilde A$, \ie we have a fiber sequence
  $
    \begin{tikzcd}[cramped,sep=small]
      \tilde A\ar[r]&A\ar[r,"g_A"]&\B\pi_1A
    \end{tikzcd}
  $.
\end{proposition}

\noindent
As an interesting immediate consequence of this result, we recover the fact that the fibers of the universal covering are the fundamental group:

\begin{proposition}
  We have a fiber sequence $\begin{tikzcd}[cramped,sep=small]\pi_1A\ar[r]&\tilde A\ar[r]&A\end{tikzcd}$.
\end{proposition}
\begin{proof}
  By \cite[Section 8.4]{hottbook}, the fiber sequence of \cref{universal-covering-fiber-sequence} can be extended on the left by $\Loop\B\pi_1A$, which is $\pi_1A$ by definition of the delooping.
\end{proof}

A careful reader could wonder why the action encoded by \cref{universal-covering-fiber-sequence} is actually the ``right'' one in the sense that it corresponds to the traditional one in topology. We can expect that there are other such actions, \ie maps $f:A\to\B\pi_1A$ with $\ker f=\tilde A$ (what is important here is that $\tilde A$ is 1-connected by \cref{ucov-1con}). In fact, it turns out is only one possible such action, up to an automorphism of~$\B\pi_1A$:

\begin{proposition}
  \label{galois-unique}
  Given a pointed connected type~$A$ and 1-connected map $f:A\to\B\pi_1A$, there is an equivalence $e:\B\pi_1A\to\B\pi_1A$ for which there is a commuting triangle
  \[
    \begin{tikzcd}
      &A\ar[dl,"g_A"']\ar[dr,"f"]&\\
      \B\pi_1A\ar[rr,"e"',"\sim"]&&\B\pi_1A
    \end{tikzcd}
  \]
\end{proposition}
\begin{proof}
  By naturality of truncation~\cite[equation (7.3.4)]{hottbook}, we have a commutative square
  \[
    \begin{tikzcd}
      A\ar[d,"f"']\ar[r,"\gtrunq-"]&\gtrunc{A}\ar[d,"\gtrunc f"]\\
      \B\pi_1A\ar[r,"\gtrunq-"']&\gtrunc{\B\pi_1A}
    \end{tikzcd}
  \]
  By definition, the upper map is $g_A$. The lower map is $g_{\B\pi_1 A}$, which is an equivalence because $\B\pi_1 A$ is a groupoid~\cite[Corollary 7.3.7]{hottbook}. The right vertical map is also an equivalence by \cite[Lemma 7.5.14]{hottbook}, because $f$ is supposed to be 1-connected. Finally, $e\defd g_{\B\pi_1 A}^{-1}\circ\gtrunc{f}$ is an equivalence as a composite of equivalences.
\end{proof}

\noindent
The situation above is a bit subtle. It states that a 1-connected map as~$f$ has to be the Galois fibration. However, this is up to an automorphism of $\B\pi_1A$, which might itself bear some information. Indeed, pointed automorphisms of $\B\pi_1 A$ correspond to group automorphisms of $\pi_1A$ which might be non-trivial. In all the applications below, insights from the corresponding constructions in algebraic topology allow us to make sure that we indeed have the right action.

\subsection{The Galois correspondence}
In algebraic topology, we have seen in \cref{ucov-quotient} that, given a nice pointed space~$A$, there is an action of its fundamental group on its universal covering. This is in fact the basis of a classification of (connected) coverings: those correspond to the subgroups of the fundamental group, see for instance~\cite[Theorem 1.38]{hatcher}. In homotopy type theory, the action is encoded by the Galois fibration (as seen with~\cref{universal-covering-fiber-sequence}) and our aim is now to develop a similar classification of coverings.
In this section, we write $\Covering(A)$ for the coverings of~$A$ which are pointed and \emph{connected}, \ie whose total space is connected. Given a group~$G$, we also write $\Subgroup(G)$ for the type of subgroups of~$G$, \ie injective maps $i:H\to G$ for some group~$H$.

The following two lemmas (Lemma~\ref{subgroup-fiber} and Lemma~\ref{subgroup-fiber}) are used in the subsequent proof of \cref{subgroup-fiber}, which will be useful to show the classification.

\renewcommand{\thelemma}{35a}
\begin{lemma}
  \label{fibers-set}
  Suppose given an injective group morphism $f : G \to H$. Then $\ker \B f$ is a set. 
\end{lemma}
\begin{proof}
  We have, by definition:
  \[
    \ker \B f \qeqdef \Sigma(x: \B G). (\pt =\B f(x))
  \]
  Given $(x,p):\ker \B f$, we want to prove that the loop space
  \[
    (x,p) = (x,p)
  \]
  is a proposition (this is enough to show that $\ker \B f$ is a set by~\cite[Theorem 7.2.1]{hottbook}). Since $\B G$ is connected, we merely have a path $q: \pt = x$. Since being a proposition for the loop space is a proposition, we can suppose given such a path $q$, and by path induction it is enough to prove that
  \[ 
    (\pt,p) =(\pt,p)
  \]
  is a proposition, with $p:\Omega \B H$.
  By the characterization of paths in $\Sigma$-types~\cite[Theorem~2.7.2]{hottbook} and transport along identity types~\cite[Theorem~2.11.3]{hottbook}, the above type is equivalent to 
  \[
    \Sigma (a: \pt =\pt).(p \cdot \ap{\B f}(a) = p)
  \]
 We then have the following chain of equivalences: 
 \begin{align*}
   \Sigma (a: \pt =\pt).&(p \cdot \ap{\B f}(a) = p)\\
   &=\Sigma (a: \star =\star).(\ap{\B f}(a) = \refl[\B f(\star)])&&\text{by left simplification by $p$}\\
   &=\Sigma (a: G).(f(a) =_H 1_H)&&\text{we use $\Omega\B G =G$ and $\Omega\B H =H$}\\
   &=\Sigma (a: G).(f(a) =_H f(1_G))&&\text{because $f$ is a morphism}\\
   &=\Sigma (a: G).(a =_G 1_G)&&\text{because $f$ is injective}
 \end{align*}
  This latter type is contractible~\cite[Lemma 3.11.8]{hottbook} and is thus a proposition.
\end{proof}

\renewcommand{\thelemma}{35b}
\begin{lemma}
  \label{fibers-quotient}
  Suppose given a group morphism $i:H\to G$ and a function $f:G\to X$, where~$X$ is a set, and $f$ is invariant by the action of $H$ on $G$ induced by right multiplication, and such that $\fib fx$ is the group~$H$ equipped with its canonical action on itself for every $x:X$. Then~$X=G/H$.
\end{lemma}
\begin{proof}
  Recall that the quotient is the set truncation of the homotopy coequalizer
  \[
    \begin{tikzcd}
      G\times H\ar[r,shift left,"\fst"]\ar[r,shift right,"\alpha"']&G\ar[r,dotted]&G\hq H
    \end{tikzcd}
  \]
  with $\fst$ the first projection and $\alpha$ the right action of~$H$ on~$G$ (we have $\alpha(a,b)\defd a\times i(b)$), \ie $G/H=\strunc{G\hq H}$. Otherwise said, $G/H$ satisfies the same universal property as the above equalizer but restricted to sets.
  Our aim is now to show that $f:G\to X$ satisfies this universal property.
  The hypothesis that $f$ is invariant under the action of~$H$ precisely means that it coequalizes the two maps.
  Suppose given another coequalizing map $g:G\to Y$ with~$Y$ a set. We need to show that there exists a unique map $h:X\to Y$ such that $h\circ f=g$:
  \[
    \begin{tikzcd}
      G\times H\ar[r,shift left,"\fst"]\ar[r,shift right,"\alpha"']&G\ar[r,"f"]\ar[dr,"g"']&X\ar[d,dotted,"h"]\\
      &&Y
    \end{tikzcd}
  \]
  The existence of such an $h$ is implied by the fact that we have, for every $x:X$, an element $a:\fib fx$ such that for every element $b:\fib fx$, we have $g(a)=g(b)$. This condition is satisfied. Indeed, under the identification of $\fib f x$ with $H$, we can take~$a$ to be the neutral element of~$H$, and the fact that the action of~$H$ on itself is transitive and that $g$ preserves this action ensures that any other image will be equal to this one.
\end{proof}

\addtocounter{lemma}{-3}
\stepcounter{lemma}

\renewcommand{\thelemma}{35}
\begin{lemma}
  \label{subgroup-fiber}
  Given a subgroup $i:H\to G$, the fiber of $\B i$ merely is the set $G/H$.
\end{lemma}
\begin{proof}
  Recall that we have
  \[
    \ker\B i
    \qdefd
    \Sigma(x:\B H).(\pt=\B i\,x)
  \]
  We define a map
  \[
    f:G\to\ker\B i
  \]
  by $f(a)\defd(\pt,p_a)$ for $a:G$ corresponding to a path $p_a:\pt=\pt$ in~$\B G$.
  The fibers of this map are
  \begin{align*}
    \fib f(x,p)
    &\defd
    \Sigma(a:G).((x,p)=(\pt,p_a))
    \\
    &=\Sigma(a:G).\Sigma(q:x=\pt).(p\cdot \ap{(\B i)}(q)=p_a)
  \end{align*}
  In particular, for $x\eqdef\pt$, we have $p\defd p_b$ for some $b:G$ and the fiber is
  \begin{align*}
    \fib f(\pt,p_b)
    &=
    \Sigma(a:G).\Sigma(c:H).(b\times i(c)=a)
    \\
    &=
    \Sigma(c:H).\Sigma(a:G).(b\times i(c)=a)
    \\
    &=H
  \end{align*}
  Since $\B H$ is connected, we can deduce that $\fib f(x,p)$ merely is~$H$ for any $(x,p):\ker\B i$. 
  Furthermore, $f$ is invariant under the action of $H$. Given $a:G$ and $b:H$, we have
  \begin{align*}
    f(b\cdot a)
    &\equiv
    (\pt, p_{(b\cdot a)})
    \\
    &\equiv 
    (\pt, p_{(a \times i(b))})
    \\
    &\equiv (\pt, p_{a}\cdot p_{i(b)})
    \\
    &= (\pt, p_{a}) & \text{by transport along $p_{i(b)}^{-1}$ in $\lambda x.(\pt =\B i(x)) $~(\cite[Th. 2.11.3]{hottbook})}
    \\
    & \equiv f(a)
  \end{align*}
  Finally, by \cref{fibers-set}, $\ker\B i$ is a set since $i$ is injective. By \cref{fibers-quotient}, we deduce that $\ker\B i$ merely is~$G/H$.
\end{proof}

\noindent
One of the main results of this paper is the following theorem, see~\cite{covering-agda} for a formalization:

\begin{theorem}
  \label{covering-classification}
  There is an equivalence between subgroups of $\pi_1(A)$ and pointed connected coverings of~$A$:
  \[
    \Subgroup(\pi_1(A))\equivto\Covering(A)
  \]
\end{theorem}
\begin{proof}
  Given a subgroup~$i:G\into\pi_1(A)$, the corresponding covering~$X$ is obtained as pullback of the Galois fibration along the delooping of group inclusion:
  \[
    \begin{tikzcd}
      C_G\ar[d,"p_G"']\ar[r]\ar[dr,phantom,pos=0,"\lrcorner"]&\B G\ar[d,"\B i"]\\
      A\ar[r,"g_A"']&\B\pi_1(A)
    \end{tikzcd}
  \]
  We need to show that $C_G$ is a connected covering, \ie it is connected and $p_G$ has $0$-truncated fibers.
  The fiber of the covering is given by pasting pullbacks
  \[
    \begin{tikzcd}
      F_G\ar[d]\ar[r]\ar[dr,phantom,pos=0,"\lrcorner"]&C_G\ar[d,"p_G"']\ar[r]\ar[dr,phantom,pos=0,"\lrcorner"]&\B G\ar[d,"\B i"]\\
      1\ar[r,"\pt"']&A\ar[r,"g_A"']&\B\pi_1(A)
    \end{tikzcd}
  \]
  and thus we have
  \[
    F_G\eqdef\ker{p_G}=\ker{\B i}=\pi_1(A)/G
  \]
  by \cref{subgroup-fiber} (to be precise, we only have the existence of such an equality, which is sufficient for our purposes exposed in next sentence).
  As a consequence, the fibers of~$p_G$ are sets. Also, by vertical pasting of pullbacks,
  \[
    \begin{tikzcd}
      \tilde A\ar[d]\ar[r]\ar[dr,phantom,pos=0,"\lrcorner"]&1\ar[d]\\
      C_G\ar[d,"p_G"']\ar[r]\ar[dr,phantom,pos=0,"\lrcorner"]&\B G\ar[d,"\B i"]\\
      A\ar[r,"g_A"']&\B\pi_1(A)
    \end{tikzcd}
  \]
  we have that the upper square is a pullback, \ie a fibration sequence of the form
  \[
    \begin{tikzcd}
      \tilde A\ar[r]&C_G\ar[r]&\B G
    \end{tikzcd}
  \]
  which describes, by action-fibration duality, an action of $G$ on $\tilde A$ such that
  \[
    C_G=\tilde A\hq G
  \]
  Therefore, $C_G$ is connected as a homotopy quotient of a connected space.

  Conversely, given a connected covering $f:X\to A$, we have
  \[
    \pi_1(f):\pi_1(X)\to\pi_1(A)
  \]
  This is a mono, because we have the long exact sequence associated to the fibration:
  \[
    \begin{tikzcd}
      \pi_1(F)\ar[r]&\pi_1(X)\ar[r,"\pi_1(f)"]&\pi_1(A)
    \end{tikzcd}
  \]
  where $F\defd\ker{f}$ is a set (because $f$ is covering) and thus $\pi_1(F)=1$.

  We now have to show that these two operations are mutually inverse. Given a connected covering $f:X\to A$, the associated connected covering (by performing the two operations) is obtained as the pullback on the left, which can be rewritten as on the right
  \[
    \begin{tikzcd}
      C_{\pi_1 X}\ar[d,"p_{\pi_1X}"']\ar[r]\ar[dr,phantom,pos=0,"\lrcorner"]&\B\pi_1X\ar[d,"\B\pi_1f"]\\
      A\ar[r,"g_A"']&\B\pi_1A
    \end{tikzcd}
    \qquad\qquad\qquad\qquad
    \begin{tikzcd}
      C_{\pi_1 X}\ar[d,"p_{\pi_1X}"']\ar[r]\ar[dr,phantom,pos=0,"\lrcorner"]&\gtrunc X\ar[d,"\gtrunc f"]\\
      A\ar[r,"\gtrunq-"']&\gtrunc A
    \end{tikzcd}
  \]
  The outer square of the diagram below commutes by naturality of the unit of the truncation, and we thus have a universal map
  $
  e:X\to C_{\pi_1X}
  $
  such that
  \[
    \begin{tikzcd}
      X\ar[ddr,bend right,"f"']\ar[drr,bend left,"\gtrunq-"]\ar[dr,"e"]&\\
      &C_{\pi_1X}\ar[d,"p_{\pi_1X}"']\ar[r]\ar[dr,phantom,pos=0,"\lrcorner"]&\gtrunc X\ar[d,"\gtrunc f"]\\
      &A\ar[r,"\gtrunq-"']&\gtrunc A
    \end{tikzcd}
  \]
  More explicitly,
  \[
    C_{\pi_1X}
    \qeqdef
    \Sigma(a:A).\Sigma(y:\gtrunc X).(\gtrunq{a}=\gtrunc{f}(y))
  \]
  and
  $
    e(x)
    =
    (f(x),\gtrunq x,\refl[\gtrunq{f(x)}])
  $.
  Because the lower-left triangle commutes, $e$ is a map between types over $A$, and thus corresponds by Grothendieck duality (\cref{grothendieck-duality}) to a family of fiberwise maps
  \begin{align*}
    e_a:\fib fa&\to\fib{p_{\pi_1X}}a\\
    (x,p)&\mapsto(f(x),\gtrunq x,\refl,p)
  \end{align*}
  indexed by $a:A$, with $p:a=f(x)$. Let us consider the case $a\eqdef\pt{}$. We have
  \[
    \fib{p_{\pi_1X}}\pt\qdefd\Sigma(a:A).\Sigma(y:\gtrunc X).\Sigma(q:\gtrunq a=\gtrunc f(y)).(\pt=a)
  \]
  We can construct an inverse map to $e_{\pt}$
  \begin{align*}
    e'_\pt:\fib{p_{\pi_1X}}\pt&\to\fib f\pt\\
    (a,y,q,p)&\mapsto(x,p\pcomp\tilde q)
  \end{align*}
  where we assume that $y=\gtrunq{x}$ because we are eliminating to $\fib f\pt$ which is a set (and thus a groupoid). Above, we have $p:\gtrunq a=\gtrunc f(x)$, which is equivalent to $\strunc{a=f(x)}$ (by~\cite[Theorem 7.3.12]{hottbook}) and we can thus suppose that $q=\strunq{\tilde q}$ with $\tilde q:a=f(x)$ because we are eliminating to a set.
  We have
  \[
    e'\circ e(x,p)
    =
    (x,p\pcomp\refl)
    =
    (x,p)
  \]
  Conversely,
  \[
    e\circ e'(a,\gtrunq x,\strunq q,p)
    =
    e(x,p\pcomp q)
    =
    (f(x),\gtrunq x,\refl,p\pcomp q)
  \]
  (we can suppose that the second and third arguments are truncations as above, because we are eliminating to a set). We thus have to show
  \[
    (f(x),\gtrunq x,\refl,p\pcomp q)
    =
    (a,\gtrunq x,\strunq q,p)
  \]
  with $p:\pt=a$ and $q:a=f(x)$.
  Abstracting over $a$, we can suppose that $q$ is $\refl$ by J, and we conclude immediately.

  Conversely, given a subgroup $i:G\into\pi_1A$, the associated subgroup (by performing the two transformations) is
  \[
    \pi_1(p_G):\pi_1C_G\to\pi_1A
  \]
  we thus have to show that $\pi_1C_G=G$ and the map~$\pi_1 p_G=i$ (note that to be precise, we should identify the sources and the targets of the maps up to equality). We have a pullback square
  \[
    \begin{tikzcd}
      C_G\ar[d,"p_G"']\ar[r]\ar[dr,phantom,pos=0,"\lrcorner"]&\B G\ar[d,"\B i"]\\
      A\ar[r,"g_A"']&\B\pi_1A
    \end{tikzcd}
  \]
  which can be extended as (see above)
  \[
    \begin{tikzcd}
      \tilde A\ar[d]\ar[r]\ar[dr,phantom,pos=0,"\lrcorner"]&1\ar[d]\\
      C_G\ar[d,"p_G"']\ar[r]\ar[dr,phantom,pos=0,"\lrcorner"]&\B G\ar[d,"\B i"]\\
      A\ar[r,"g_A"']&\B\pi_1(A)
    \end{tikzcd}
  \]
  The map $C_G\to\B G$ is 1-connected because the fiber is $\tilde A$ which is 1-connected by \cref{n-cov-conn}. By \cite[Lemma 7.5.14]{hottbook}, it thus induces an equivalence $\gtrunc{C_G}\equivto\gtrunc{\B G}$, and thus $\pi_1(C_G)=\pi_1(\B G)=G$. We should have the fact that $\pi_1p_G=i$ similarly, by applying $\pi_1$ to the above square (which is not anymore a pullback but remains commutative).
\end{proof}

\begin{example}
  Consider the case $A\eqdef\S1$. The associated Galois fibration is
  \[
    \begin{tikzcd}
      \S1\ar[r,"\gtrunq-"]&\gtrunc{\S1}\eqdef\B\pi_1\S1\equivto\B\Z
    \end{tikzcd}
  \]
  and is an equivalence, because $\S1$ is a groupoid. The subgroups of $\Z$ are of the form $i_n:\Z\to\Z$ with $i_n(k)=n\times k$ with $n>0$ or $i_0:0\to\Z$. And thus (connected) coverings of the circle are obtained as pullbacks of the form
  \[
    \begin{tikzcd}
      C_n\ar[d,"p_n"']\ar[r]\ar[dr,phantom,pos=0,"\lrcorner"]&\B\Z\ar[d,"\B i_n"]\\
      \S1\ar[r,"g_{\S1}"']&\B\Z
    \end{tikzcd}
  \]
  Since the lower arrow is an equivalence, the pullback $C_n$ is $\B\Z$, \ie $\S1$. And $p_n:\S1\to\S1$ is the pointed map of degree $n$, sending the loop to loop$\null^n$.
\end{example}


\section{Applications}
\label{applications}
\subsection{Coverings of lens spaces}
Lens spaces were defined in homotopy type theory by the authors of this paper in~\cite{lens}. We briefly recall here their definition. Given natural numbers $l$ and $n$ with $l$ relatively prime to $n$, we write
\[
  \phi_n^l:\S1\to\B\Z_n
\]
for the pointed map sending $\Sloop$ to the loop in $\B\Z_m$ corresponding to $l:\Z_m$. We have a fiber sequence
\[
  \begin{tikzcd}
    \S1\ar[r]&\S1\ar[r,"\phi_n^l"]&\B\Z_n
  \end{tikzcd}
\]
which is obtained by delooping the exact sequence
$
  \begin{tikzcd}[cramped]
    \Z\ar[r,"-\times n"]&\Z\ar[r,"-\times l"]&\B\Z_n
  \end{tikzcd}
$,
see \cite[Section 6.2]{lens}.

\begin{definition}
  Given a sequence $l_1,\ldots,l_k$ of natural numbers all prime with $n$, the associated \emph{lens space} $L_n^{l_1,\ldots,l_k}$ is the source of the map
  \[
    \phi_n^{l_1}\join\ldots\join \phi_n^{l_k}:L_n^{l_1,\ldots,l_k}\to\B\Z_n
  \]
  which we simply write as $\phi_n^{l_1,\ldots,l_k}$ in the following.
\end{definition}

\noindent
It can be shown that the fiber of this map is $\S{2k-1}$~\cite[Section 6.2]{lens}, which is thus $2k-2$ connected. In particular, we have that $\pi_1L_n^{l_1,\ldots,l_k}=\pi_1\B\Z_n=\Z_n$.

We now classify covering spaces of lens spaces. We have seen in \cref{covering-classification} that they correspond to subgroups of $\Z_n$, which are the $\Z_m$ with $m\divides n$. Given such a $\Z_m$, the corresponding covering~$C_m$ is obtained by taking the inclusion $i_m:\Z_m\to\Z_n$, delooping it, and pulling back along the Galois fibration:
\[
  \begin{tikzcd}[column sep=large]
    C_m\ar[d,"p_m"']\ar[r]\ar[dr,phantom,pos=0,"\lrcorner"]&\B\Z_m\ar[d,"\B i_m"]\\
    L_n^{l_1,\ldots,l_k}\ar[r,"g_{L_n^{l_1,\ldots,l_k}}"']&\B\Z_n
  \end{tikzcd}
\]
In order to perform computations, we first note that, by \cref{galois-unique}, we can replace the map $g_{L_n^{l_1,\ldots,l_k}}$ at the bottom by $\phi_n^{l_1,\ldots,l_k}$ (up to an automorphism of the target).



\begin{proposition}
  \label{lens-pullback}
  For natural numbers $l,m,n,p$ with $l$ prime with $n$, and $n=mp$, we have a pullback square
  \[
    \begin{tikzcd}
      \S1\ar[d,"s_p"']\ar[r,"\phi^l_m"]\ar[dr,phantom,pos=0,"\lrcorner"]&\B\Z_m\ar[d,"\B i_m"]\\
      \S1\ar[r,"\phi_n^l"']&\B\Z_n
    \end{tikzcd}
  \]
  where the vertical map~$s_p$ sends $\Sloop$ to $\Sloop^p$.
\end{proposition}
\begin{proof}
  We write~$X$ for pullback of $\phi_n^l$ and $\B i_m$ as shown on the left below:
  \[
    \begin{tikzcd}[column sep=large]
      X\ar[d,"(\phi_n^l)^*\B i_m"']\ar[r,"(\B i_m)^*\phi_n^l"]\ar[dr,phantom,pos=0,"\lrcorner"]&\B\Z_m\ar[d,"\B i_m"]\\
      \S1\ar[r,"\phi_n^l"']&\B\Z_n
    \end{tikzcd}
    \qquad\qquad\qquad\qquad
    \begin{tikzcd}
      \Loop X\ar[d]\ar[r]\ar[dr,phantom,pos=0,"\lrcorner"]&\Z_m\ar[d,"i_m"]\\
      \Z\ar[r,"-\times l"']&\Z_n
    \end{tikzcd}
  \]
  The type $X$ is pointed (because both maps $\phi_n^l$ and $\B i_m$ are) and connected (this can be shown as in \cite[Lemma 26]{lens}). Finally, since loop spaces commute to pullbacks because they are right adjoints, we have a pullback square of groups as on the right above. From this, we can deduce that $\Loop X\isoto\Z$ as follows. The preceding pullbacks means that we have
  \[
    \Loop X
    \quad=\quad
    \Sigma((a,b):\Z\times\Z_m).(al=_{\Z_n}bp)
  \]
  We write $f:\Z\to\Loop X$ for the map sending $1$ to $(p,l)$, and we claim that this is an isomorphism.

  First, $f$ is injective because $(p,l)$ is free. Indeed, suppose that $f(x)\eqdef(xp,xl)=0$. In particular, $xp=_\Z 0$ and since $p\neq 0$ we have $x=0$.
  Second, $f$ is surjective. Fix $(a,b):\Loop X$. There is $y\in\Z$ such that $al=bp+yn=p(b+ym)$. Therefore $p\divides al$, but $l$ is prime with $n$ and $p\divides n$ so that $l$ is prime with $p$ and thus $p\divides a$. There thus exists $x$ such that $a=px$. Thus $pxl-bp=_{\Z_n}0$, \ie $p(xl-b)=_{\Z_n}0$, that is $mp\defd n\divides p(xl-b)$ and thus $m\divides xl-b$, \ie $b=_{\Z_m}xl$. Finally, $(a,b)=(xp,xl)\eqdef f(x)$ and $f$ is surjective.
  The maps of the pullback are the projections from~$\Loop X$ and are thus the expected ones.
\end{proof}

\begin{theorem}
  For $l_1,\ldots,l_k$ natural numbers prime with $n$, and $k>1$, the connected coverings of $L_n^{l_1,\ldots,l_k}$ are precisely the $L_m^{l_1,\ldots,l_k}$ with $m\divides n$, and the projections maps are given by $s_p\join_{\B\Z_n}\ldots\join_{\B\Z_n}s_p$.
\end{theorem}
\begin{proof}
  By \cref{covering-classification}, the coverings of $L_n^{l_1,\ldots,l_k}$ correspond to subgroups of $\pi_1L_n^{l_1,\ldots,l_k}$, \ie the subgroups of $\Z_n$, and those are of the form $\Z_m$ for $m\divides n$. More precisely, by \cref{covering-classification}, given such a subgroup $i_m:\Z_m\to\Z_n$, the corresponding covering space is given by the pullback of the delooping this map along the Galois fibration:
  \[
    \begin{tikzcd}[column sep=large]
      X\ar[d]\ar[r]\ar[dr,phantom,pos=0,"\lrcorner"]&\B\Z_m\ar[d,"\B i_m"]\\
      L_n^{l_1,\ldots,l_k}\ar[r,"g_{L_n^{l_1,\ldots,l_k}}"']&\B\Z_n
    \end{tikzcd}
  \]
  As remarked in~\cref{galois-unique}, we can replace the bottom map by $\phi_n^{l_1,\ldots,l_k}$. By \cite[Theorem 25]{lens}, pullback commute with joins of maps, so the map $X\to \B\Z_m$ can be computed as the iterated join of the maps obtained by pulling back~$\phi_n^{l_i}$ along $\B i_m$, which, by \cref{lens-pullback}, are the $\phi_m^{l_i}:\S1\to\B\Z_m$. The pullback of $\phi_n^{l_1,\ldots,l_k}$ along $\B i_m$ is thus $\phi_m^{l_1,\ldots,l_k}$. Similarly, by \cite[Theorem 25]{lens} and \cref{lens-pullback}, the vertical map $X\to L_n^{l_1,\ldots,l_k}$ is obtained as the join of $n$ instances of $s_p$. Finally, we obtain the pullback square
    \[
    \begin{tikzcd}[column sep=large]
       L_m^{l_1,\ldots,l_k}\ar[d,"s_p^{\join_{\B\Z_n}^k}"']\ar[r,"\phi_m^{l_1,\ldots,l_k}"]\ar[dr,phantom,pos=0,"\lrcorner"]&\B\Z_m\ar[d,"\B i_m"]\\
      L_n^{l_1,\ldots,l_k}\ar[r,"g_{L_n^{l_1,\ldots,l_k}}"']&\B\Z_n
    \end{tikzcd}
    \qedhere
  \]
\end{proof}

\subsection{Constructing the hypercubical manifold and the homology sphere}
We would like to illustrate here another kind of situation where the Galois fibration occurs when constructing types corresponding to well-known spaces, albeit in a somewhat hidden way. The general idea is the following. Suppose that we have a group~$G$ and we want to define a coherent action of~$G$ on a type~$A$, which is not supposed to be truncated (in particular, it might not be a set). This amounts to define a map $\psi:\B G\to\U$ such that $\psi\,\pt=A$, which requires eliminating to a type which is not a groupoid, and is thus difficult to perform directly. But we can use the action-fibration duality, which brings a fresh point of view on the problem. Indeed, by the Grothendieck duality, provided we can construct the homotopy quotient $\ol A$ of $A$ under the action of~$G$, constructing the action amounts to defining a map $\phi:\ol A\to\B G$. However, it is not clear which map this should be. When the type~$A$ is simply connected, we know that $\phi$ has to be the Galois fibration by \cref{galois-unique}.

\subparagraph{The hypercubical manifold.}
\ifx\authoranonymous\relax It was\else We have\fi{} defined and studied in~\cite{hypercubical} a type corresponding to the \emph{hypercubical manifold}. Topologically, this manifold~$K$ was defined by Poincaré as a space obtained by identifying the opposite faces of a cube after a quarter-turn rotation:
\[
  \begin{tikzcd}[sep=small]
    &\ar[dl,"w"']a\ar[dd,dotted,"x",near start]\ar[rr,"y"]&&b\\
    b&&\ar[ll,"z",near start]a\ar[dd,"w",near start]\ar[ur,"x"]\\
    &b&&\ar[ll,dotted,"w",near end]\ar[dl,"y"]a\ar[uu,"z"']\\
    a\ar[uu,"y"]\ar[ur,dotted,"z"]\ar[rr,"x"']&&b
  \end{tikzcd}
\]
The fundamental group of this space is the quaternion group~$Q$ and its universal covering is the 3-sphere~$\S3$. Thus, $K$ can also be obtained as a quotient of~$\S3$ under~$Q$.

In homotopy type theory, a type corresponding to~$K$ is easily defined as a higher inductive type. However, in order to work with it and validate its construction, we need to show that it can be obtained as a quotient of $\S3$ under the action of~$Q$, \ie from a map $\B Q\to\U$ such that the image of~$\pt$ is $\S3$. However, such a map is difficult to construct directly using the elimination principle of $\B Q$ because $\S3$ is not $n$-truncated for any $n$. As explained above, we can instead adopt a fibrational point of view and construct a map $\phi:K\to\B Q$ whose fiber is~$\S3$, \ie a fiber sequence
\[
  \begin{tikzcd}
    \S3\ar[r]&K\ar[r,"\phi"]&\B Q
  \end{tikzcd}
\]
thus showing that $K$ is a quotient of~$\S3$ under an action of~$Q$.
Details can be found in~\cite{hypercubical}.

\subparagraph{The homology sphere.}
When investigating the notion of homology, it was not clear at first whether homology would be fine enough in order to characterize homotopy types. It turns out that this is not the case, which was first shown by Poincaré by introducing a space, the \emph{homology sphere} (also known as the \emph{Poincaré dodecahedral space}), which has the same homology type as the sphere $\S3$, but not the same homotopy type~\cite{poincare1904cinquieme}.


Following~\cite{analysis-situs-dodeca}, we define the homology sphere $D$ as a higher inductive type corresponding to a dodecaedron where each face is identified with the opposite one after a rotation of a fifth-turn. The classes of 0-, 1- and 2-cells respectively have 4, 3 and 2 elements, and edges are identified as indicated by the colors below:
\begin{center}
  \tdplotsetmaincoords{70}{70}
\begin{tikzpicture}[tdplot_main_coords,thick,scale=0.8]
\tikzmath{\x=(1+sqrt(5))/2;}
\coordinate (A) at (1,0,0);
\coordinate (B) at ({cos(72)},{sin(72)},0);
\coordinate (C) at ({-cos(36)},{sin(36)},0);
\coordinate (D) at ({-cos(36)},{-sin(36)},0);
\coordinate (E) at ({cos(72)},{-sin(72)},0);
\coordinate (F) at (\x,0,1);
\coordinate (G) at ({\x*cos(36)},{\x*sin(36)},\x);
\coordinate (H) at ({\x*cos(72)},{\x*sin(72)},1);
\coordinate (I) at ({-\x*cos(72)},{\x*sin(72)},\x);
\coordinate (J) at ({-\x*cos(36)},{\x*sin(36)},1);
\coordinate (K) at (-\x,0,\x);
\coordinate (L) at ({-\x*cos(36)},{-\x*sin(36)},1);
\coordinate (M) at ({-\x*cos(72)},{-\x*sin(72)},\x);
\coordinate (N) at ({\x*cos(72)},{-\x*sin(72)},1);
\coordinate (O) at ({\x*cos(36)},{-\x*sin(36)},\x);
\coordinate (P) at (-1,0,\x+1);
\coordinate (Q) at ({-cos(72)},{-sin(72)},\x+1);
\coordinate (R) at ({cos(36)},{-sin(36)},\x+1);
\coordinate (S) at ({cos(36)},{sin(36)},\x+1);
\coordinate (T) at ({-cos(72)},{sin(72)},\x+1);

\draw[blue,dotted]
(L) -- (K)
;
\draw[brown,dotted]
(L) -- (D)
;
\draw[cyan,dotted]
(C) -- (D)
;
\draw[green,dotted]
(J) -- (C)
;
\draw[magenta,dotted]
(K) -- (J)
;
\draw[olive,dotted]
(M) -- (L)
;
\draw[orange,dotted]
(D) -- (E)
;
\draw[red,dotted]
(C) -- (B)
;
\draw[black,dotted]
(J) -- (I)
;
\draw[teal,dotted]
(P) -- (K)
;

\draw[blue]
(O) -- (R)
(B) -- (H)
;
\draw[brown]
(O) -- (F)
(T) -- (I)
;
\draw[cyan]
(F) -- (G)
(P) -- (Q)
;
\draw[green]
(S) -- (G)
(M) -- (N)
;
\draw[magenta]
(R) -- (S)
(E) -- (A)
;
\draw[olive]
(A) -- (B)
(S) -- (T)
;
\draw[orange]
(Q) -- (R)
(I) -- (H)
;
\draw[red]
(P) -- (T)
(O) -- (N)
;
\draw[black]
(Q) -- (M)
(A) -- (F)
;
\draw[teal]
(N) -- (E)
(G) -- (H)
;
\end{tikzpicture}

\end{center}
The resulting space has one 3-cell, 6 2-cells, 10 1-cells and 5 0-cells. The fundamental group~$\pi_1(D)$ is the group of order 120, known as the \emph{binary icosahedral group}, which can be presented as
$
\pres{r,s,t}{r^2=s^3=t^5=rst}
$. 
To exhibit $D$ as a quotient of the $3$-sphere under the action of the binary icosagedral group, we need to show that we have a fibration
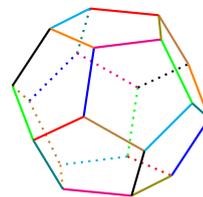
\[
  \begin{tikzcd}
    \S3\ar[r]&D\ar[r,"\phi"]&\B\pi_1(D)
  \end{tikzcd}
\]
As for the hypercubical manifold, we can build a model for the fundamental group of $D$ in two steps: we first construct the fundamental groupoid (whose 0-, 1- and 2-cells are generated by the 0-, 1- and 2-cells involved in the description of~$D$ as a cellular complex), and then we contract 1-cells in order to obtain a type with only one 0-cell, which is thus a delooping of a group, \ie $\B \pi_1(D)$. The map~$\phi$ is then the canonical one induced by this process. Finally, since~$D$ is constructed by attaching cells, \ie has a canonical description as a colimit, we can use the flattening lemma in order to compute its fiber, which we claim to be $\S3$. This will be detailed in subsequent works.

\section{Future work}
\label{conclusion}
We have defined and studied $n$-covering types, and formalized their classification for $n=0$. In the future, we would like to provide an explicit construction of covering types of many interesting and natural types (such as the hypercubical manifold and the homology sphere presented above). In passing, we would like to mention here that \ifx\authoranonymous\relax it was shown\else we show\fi{} in~\cite{generated,presented} that Cayley complexes are universal coverings (and Cayley graphs are universal $(-1)$-coverings). A natural question is also whether the Galois correspondence can be extended to higher coverings. Its exploration is left for future work.

\bibliography{papers}

\end{document}